\documentclass[prd,reprint, nofootinbib,amsmath,amssymb, aps, floatfix]{revtex4-1}
\usepackage{dcolumn}

\usepackage{bm}
\pdfoutput=1
\usepackage{amssymb,amsmath,amsthm,graphicx,ulem}
\usepackage{hyperref}
\usepackage{graphicx,subfigure}
\usepackage{epsfig}
\usepackage{amsfonts}
\usepackage[usenames]{color}
\usepackage{enumerate}
\usepackage[margin=1.99cm]{geometry}
\usepackage[T1]{fontenc}

\newtheorem{thm}{Theorem}
 
\newtheorem{lem}[thm]{Lemma}
\theoremstyle{remark}
\newtheorem{defn}[thm]{Definition}
\newtheorem{rem}[thm]{Remark}

\usepackage{float}

\usepackage{xcolor}

\newcommand{\ccirc}{\mathbin{\mathchoice
  {\xcirc\scriptstyle}
  {\xcirc\scriptstyle}
  {\xcirc\scriptscriptstyle}
  {\xcirc\scriptscriptstyle}
}}
\newcommand{\xcirc}[1]{\vcenter{\hbox{$#1\!\circ\!$}}}

\newcommand{\be}{\begin{equation}}
\newcommand{\ee}{\end{equation}}
\newcommand{\Ref}[1]{Ref.~\cite{#1}}
\newcommand{\Fig}[1]{Fig.~\ref{#1}}
\newcommand{\Eq}[1]{Eq.~\eqref{#1}}

\newcommand{\Sec}[1]{Sec.~\ref{#1}}

\begin{document}
\normalem
\title{The Boundary of the Future}
\author{Chris Akers, Raphael Bousso, Illan F. Halpern, and Grant N. Remmen}
 \email{cakers@berkeley.edu\\bousso@lbl.gov\\illan@berkeley.edu\\grant.remmen@berkeley.edu}
\affiliation{Center for Theoretical Physics and Department of Physics\\
University of California, Berkeley, CA 94720, USA}
\affiliation{Lawrence Berkeley National Laboratory, Berkeley, CA 94720, USA}
\bibliographystyle{utphys-modified}

\begin{abstract}
We prove that the boundary of the future of a surface $K$ consists precisely of the points $p$ that lie on a null geodesic orthogonal to $K$ such that between $K$ and $p$ there are no points conjugate to $K$ nor intersections with another such geodesic. Our theorem has applications to holographic screens and their associated light sheets and in particular enters the proof that holographic screens satisfy an area law.
\end{abstract}

\maketitle

\section{Theorem}
\label{intro}

In this paper, we prove the following theorem establishing necessary and sufficient conditions for a point to be on the boundary of the future of a surface in spacetime. (An analogous theorem holds for the past of $K$.)

\begin{thm}\label{Theorem}
Let $(M, g)$ be a smooth,\footnote{Nowhere in the proof will more than two derivatives be needed, so the assumption of smoothness for $M$ and $K$ can be relaxed everywhere in this paper to $C^2$. } globally hyperbolic spacetime and let $K$ be a smooth codimension-two submanifold of $M$ that is compact and acausal. Then a point $b\in M$ is on the boundary of the future of $K$ if and only if all of the following statements hold: 
\begin{enumerate}[(i)]
\item $b$ lies on a future-directed null geodesic $\gamma$ that intersects $K$ orthogonally.
\item $\gamma$ has no points conjugate to $K$ strictly before $b$.
\item $\gamma$ does not intersect any other null geodesic orthogonal to $K$ strictly between $K$ and $b$.
\end{enumerate}
\end{thm}

Theorem~\ref{Theorem} enumerates the conditions under which a light ray, launched normally from a surface, can exit the boundary of the future of that surface and enter its chronological future. In essence, this happens only when the light ray either hits another null geodesic launched orthogonally from the surface or when the light ray encounters a caustic, in a sense that will be made precise in terms of special conditions on the deviation vectors for a family of infinitesimally-separated geodesics. These two possibilities for the fate of the light ray are illustrated in \Fig{fig:par}.
\begin{figure}[h]
\includegraphics[width=0.8\columnwidth]{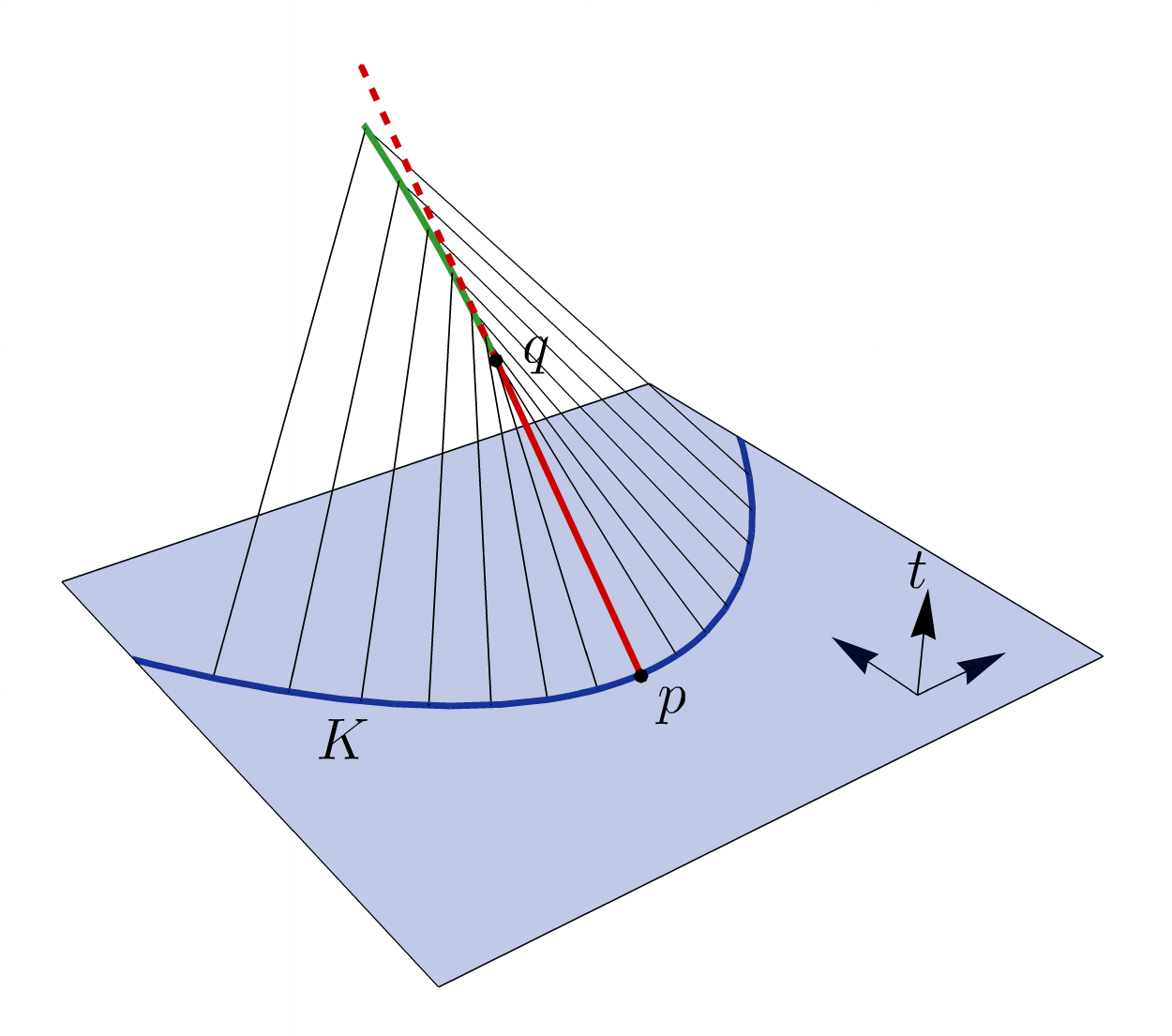}
\caption{Possibilities for how a null geodesic orthogonal to a surface can exit the boundary of its future. In this example, a parabolic surface $K$ (blue line) lies in a particular spatial slice. A future-directed null geodesic (red line) is launched orthogonally from $p$. At $q$, it encounters a caustic, entering the interior of the future of $K$ (red dashed line).  The point $q$ is conjugate to $K$. Other null geodesics orthogonal to $K$ (black lines) encounter nonlocal intersections with other such geodesics along the green line, where they exit the boundary of the future of $K$.}\label{fig:par}
\end{figure}

The theorem is useful for characterizing the causal structure induced by spatial surfaces. In particular, if $K$ splits a Cauchy surface into two parts, then Theorem~\ref{Theorem} implies that the four orthogonal null congruences fully characterize the associated split of the spacetime into four portions: the future and past of $K$ and the domains of dependence of each of the two spatial sides (see Fig.~\ref{fig-split}). This is of particular interest when $K$ is a holographic screen~\cite{CEB2}. Then some of the orthogonal congruences form {\em light sheets}~\cite{CEB1} such that the entropy of matter on a light sheet is bounded by the area of $K$. This relation makes precise the notion that the universe is like a ``hologram''~\cite{Tho93,Sus95,FisSus98} and should be described as such in a quantum gravity theory. Such holographic theories have indeed been identified for a special class of spacetimes~\cite{Mal97}.

Specifically, Theorem~\ref{Theorem} plays a role in the recent proof of a novel area theorem for holographic screens~\cite{BouEng15a,BouEng15b}, where it was assumed without proof. It also enters the analogous derivation of a related Generalized Second Law in cosmology~\cite{BouEng15c} from the Quantum Focusing Conjecture~\cite{BouFis15a}.

Although our motivation lies in applications to General Relativity and quantum gravity, we stress that the theorem itself is purely a statement about Lorentzian geometry. It does not assume Einstein's equations and so in particular does not assume any conditions on the stress tensor of matter.

\begin{figure}[h]
\includegraphics[width=0.75\columnwidth]{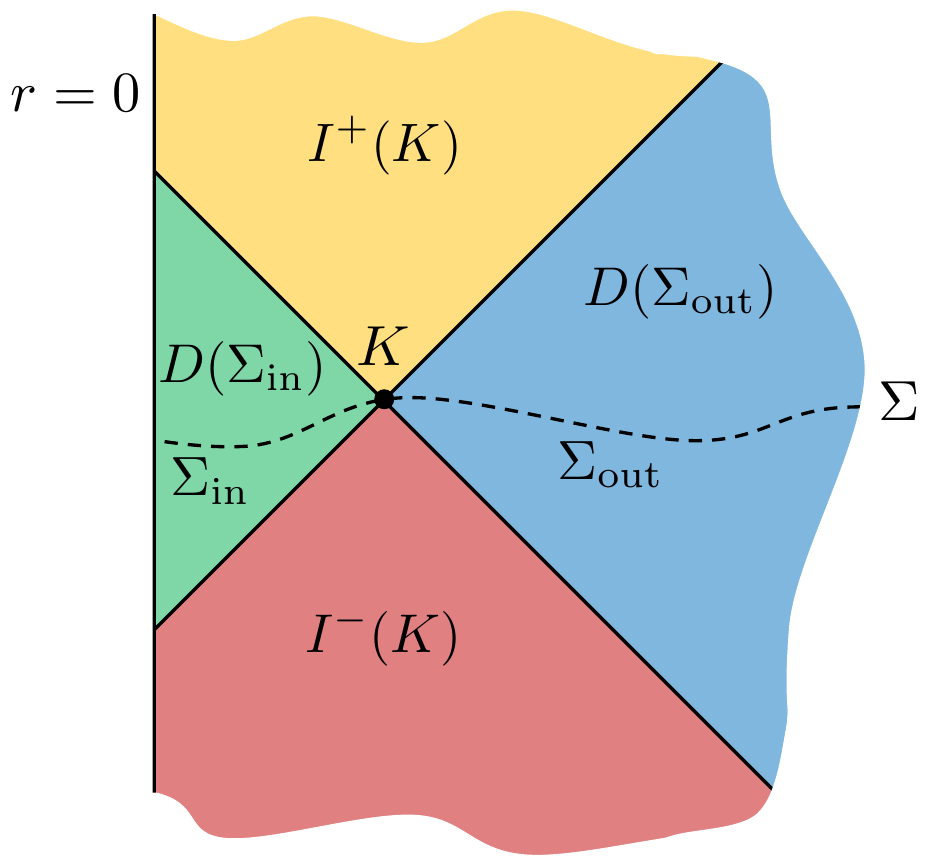}
\caption{In this generic Penrose diagram, the codimension-two surface $K$ (black dot) splits a Cauchy surface $\Sigma$ (dashed line) into two parts $\Sigma_{\rm in}$, $\Sigma_{\rm out}$. This induces a splitting of the spacetime $M$ into four parts: the past and future of $K$ (red, yellow) and the domains of dependence of $\Sigma_{\rm in}$ and $\Sigma_{\rm out}$ (green, blue)~\cite{BouEng15b}. Theorem~\ref{Theorem} guarantees that this splitting is fully characterized by the four orthogonal null congruences originating on $K$ (black diagonal lines).}\label{fig-split}
\end{figure}

\paragraph*{Related Work.} Parts of the ``only if'' direction of the theorem are a standard textbook result~\cite{Wald}, except for $(iii)$, which we easily establish. The  ``if'' direction is nontrivial and takes up the bulk of our proof.

\Ref{Beem} considers the {\it cut locus}, i.e., the set of all {\it cut points} associated with geodesics starting at some point $p\in M$. Given a geodesic $\gamma$ originating at $p$, a future null cut point, in particular, can be defined in terms of the Lorentzian distance function or equivalently as the final point on $\gamma$ that is in the boundary of the future of $p$. As shown in Theorem~5.3 of \Ref{Beem}, if $q$ is the future null cut point on $\gamma$ of $p$, then either $q$ corresponds to the first future conjugate point of $p$ along $\gamma$, or another null geodesic from $p$ intersects $\gamma$ at $q$, or both. Our theorem can be viewed as an analogous result for geodesics orthogonal to codimension-two surfaces and a generalization of our theorem implies the result of \Ref{Beem} as a special case. The codimension-two surfaces treated by our theorem are of significant physical interest due to the important role of holographic screens in the study of quantum gravity (see, e.g., Ref.~\cite{NettaAron2017} for very recent results on the coarse-grained black hole entropy). We encountered nontrivial differences in proving the theorem for surfaces. Moreover our condition ($ii$) places stronger constraints on the associated deviation vector, as we discuss in \Sec{sec-jac}.\footnote{After this paper first appeared, we were made aware of Refs.~\cite{X,Y}, which also generalize the results of \Ref{Beem} to codimension-two surfaces. Our work goes further in that we more strongly constrain the type of conjugacy to be that of Def.~\ref{def-conjtoK}. This is crucial for making contact with the notion of points ``conjugate to a surface'' used in the physics literature, e.g., in \Ref{Wald}.}

The previously known parts of the ``only if'' direction of Theorem~\ref{Theorem} were originally established in the context of proving singularity theorems  \cite{Penrose2, Hawking2}.  It would be interesting to see whether Theorem~\ref{Theorem} can be used to derive new or stronger results on the formation or the cosmic censorship of spacetime singularities.

\paragraph*{Generalizations.} 

As we are only concerned with the causal structure, the metric can be freely conformally rescaled. Thus, a version of Theorem~\ref{Theorem} still holds for noncompact $K$, as long as it is compact in the conformal completion of the spacetime, i.e., in a Penrose diagram. A situation in which this may be of interest is for surfaces anchored to the boundary of anti-de Sitter space.

Furthermore, the theorem can be generalized to surfaces of codimension other than two, but in that case we can say less about the type of conjugate point that orthogonal null geodesics may encounter. We will discuss this further in \Sec{proof}.

\paragraph*{Notation.} Throughout, we use standard notation for causal structure. A causal curve is one for which the tangent vector is always timelike or null. The causal (respectively, chronological) future of a set $S$ in our spacetime $M$, denoted by $J^+(S)$ (respectively, $I^+(S)$) is the set of all $q\in M$ such that there exists $p\in S$ for which there is a future-directed causal (respectively, timelike) curve in $M$ from $p$ to $q$. For the past ($I^-(S)$, $J^-(S)$, etc.), similar definitions apply. We will denote the boundary of a set $S$ by $\dot{S}$. Standard results \cite{Wald} include that $I^\pm(S)$ is open and that $\dot{J}^\pm(S)=\dot{I}^\pm(S)$. We will call a set $S$ {\em acausal} if there do not exist distinct $p,q\in S$ for which there is a causal path in $M$ from $p$ to $q$. A spacetime is said to be {\em globally hyperbolic} if it contains no closed causal curves and if $J^+(p)\cap J^-(q)$ is compact for all $p,q\in M$. Equivalently~\cite{Geroch}, $M$ has the topology of $\Sigma \times \mathbb{R}$ for some {\em Cauchy surface} $\Sigma$; that is, $\Sigma$ is a surface for which, for all $p\in M$, every inextendible timelike curve through $p$ intersects $\Sigma$ exactly once.

\paragraph*{Outline.} In \Sec{def}, we review the notion of a conjugate point and establish some useful lemmas. In \Sec{proof}, we prove Theorem~\ref{Theorem}.

\section{Conjugate Points to a Surface}
\label{def}

\subsection{Exponential Map}

Let $(M,g)$ be a smooth, globally hyperbolic spacetime of dimension $n>2$. Thus, $M$ is a manifold with metric $g$ of signature $(-,+,\ldots,+)$. (As already noted, we will be concerned only with the causal structure of $M$, so $g$ need only be known up to conformal transformations.)

For $p\in M$, let $T_pM$ be the tangent vector space at $p$ and let $TM \equiv \bigcup_{p\in M} \{p\}\times T_pM$ be the tangent bundle of $M$. $TM$ has a natural topology that makes it a manifold of dimension $2n$. In the open subsets associated with charts of $M$, $TM$ is diffeomorphic to open subsets of $\mathbb{R}^{2n},$ corresponding to $n$ coordinates for the location of $p\in M$ and $n$ components of a tangent vector $v\in T_pM$.  The tangent space of $TM$ at $(p,v)$ is
\be
T_{p,v}TM = T_pM \times T_vT_pM.
\ee
For every $(p,v)\in TM$, there is a unique inextendible geodesic,
\be 
 c_{p,v}: (a,b) \to M, ~s \mapsto c_{p,v}(s),
\ee
where $a,b \in \mathbb{R}\cup\{-\infty,\infty\}$,  with affine parameter $s$ and tangent vector $v \in T_pM$ given by the pushforward of $d/ds$ by $c_{p,v}$ at the point $p=c_{p,v}(0)\in M$. It is convenient to include the degenerate curves obtained with $v=0$.

\begin{defn}
The {\em exponential map} is defined by:\footnote{If the spacetime is not geodesically complete, the exponential map can only be defined on the subset of $TM$ consisting of the $(p,v)$ such that $c_{p,v}$ can be extended to $\lambda=1.$  This restriction will be left implicit in this paper.}
\be 
\exp\, :~ TM \to M, ~(p,v)\mapsto c_{p,v}(1).
\ee
\end{defn}

Restrictions of $\exp$ to submanifolds of $TM$ are frequently of interest. To study the congruence of geodesics emanating from a given point, one may restrict to $\exp_p:\,T_pM\to M\, , ~ v\mapsto c_{p,v}(1)$. Moreover, one can define the differential of $\exp_p$, $\exp_{p*}:T_v T_p M\rightarrow T_{c_{p,v}(1)}M$, which describes how $\exp_p v$ varies due to small changes in $v$. See \Fig{fig:exp} for an illustration of the exponential map and its differential. In this paper, we will consider a different restriction suited to the study of the geodesics orthogonal to a given spatial surface; we will define the differential in more detail for this restriction below. 

\begin{figure}[h]
\includegraphics[width=0.5\columnwidth]{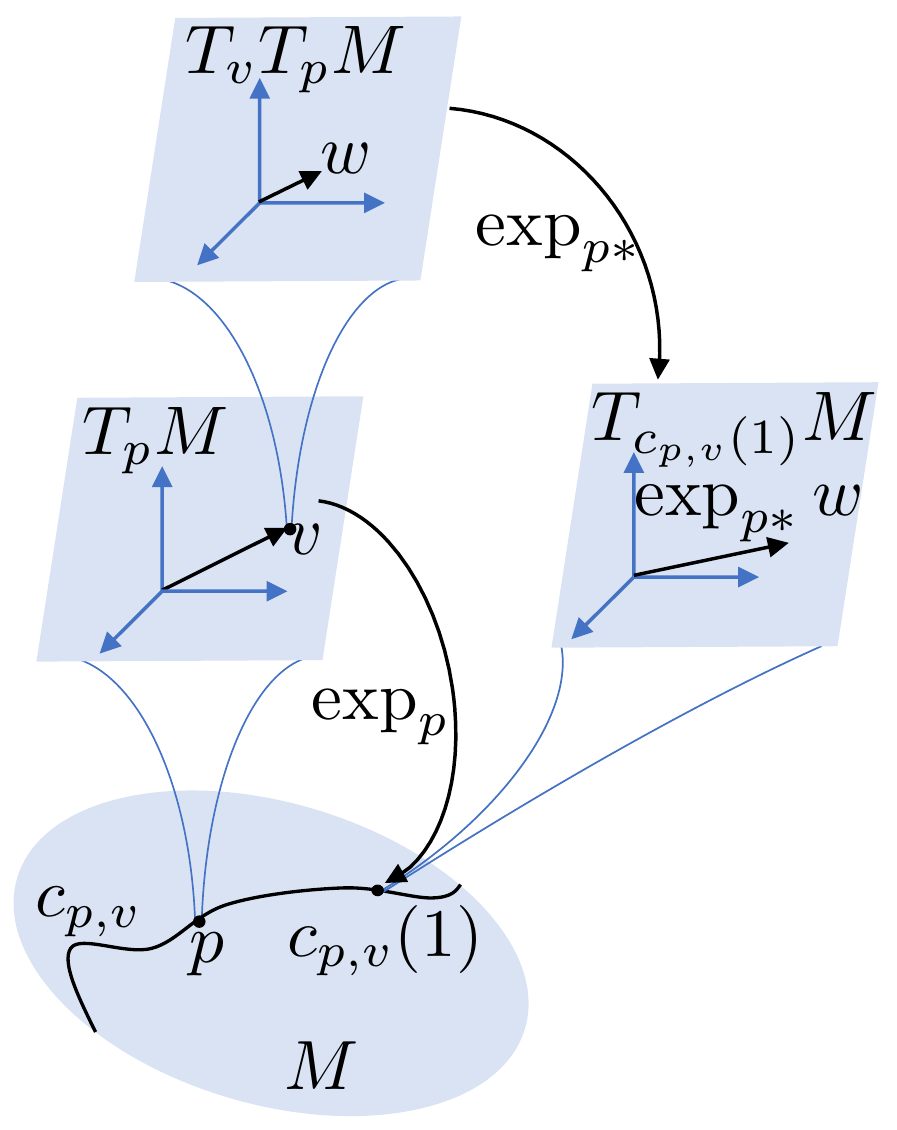}
\caption{An illustration of the exponential map $\exp$, which takes a vector in $TM$ to a point in $M$, and the Jacobian of the exponential map, which takes a vector in the tangent space $TTM$ of $TM$ to a vector in $TM$. }\label{fig:exp}
\end{figure}

Let $K\subset M$ be a smooth submanifold.  We consider the normal bundle 
$$NK\equiv\bigcup_{p\in K}\{p\}\times T_p K^\perp,$$ 
where $ T_p K^\perp$ is the two-dimensional tangent vector space perpendicular to $K$ at $p$. The normal bundle has the structure of an $n$-dimensional manifold. Its tangent space at $(p,v)\in NK$ is
\begin{equation}
T_{p,v}NK=T_p K\times T_v T_p K^\perp.
\label{eq-factor}
\end{equation} 
Here, $T_p K$ is the tangent space of $p$ in the manifold $K$; that is, $T_pK$ is the subspace of $T_p M$ normal to $T_pK^\perp$. Note that $T_p K$ is of the same dimension as $K$.
\begin{defn}
   The {\em surface-orthogonal exponential map}
  \be 
  \exp_K\,:~ NK\to M\, , ~(p,v)\mapsto c_{p,v}(1)
  \ee
  is the restriction of $\exp$ to $NK$. 
\end{defn}
\begin{defn}
  The {\em Jacobian} or {\em differential} of the exponential map is given by
 \be 
 \exp_{K*}\,:~T_{p,v}NK \to TM\, ,~ w \mapsto \exp_{K*} w.
 \ee
It is a linear map between vectors that captures the response of $\exp_K$ to small variations in its argument. It is defined by requiring that $(\exp_{K*}w)(f)=w(f\ccirc\exp_K)$ for any function $f:\,M\to \mathbb{R}$. Note $\exp_{K*} w$ is the {\em pushforward} of $w$ by $\exp_K$. If $x^\alpha$ are coordinates in an open neighborhood of $(p,v) \in NK$ and $y^\beta$ are coordinates in an open neighborhood of $\exp_K (p,v)\in M$ and we write the vectors in coordinate form, $w= \sum w^\alpha (\partial/\partial x^\alpha)$ and $\exp_{K*} w = \sum \hat w^\beta (\partial/\partial y^\beta)$, then the components are related by the {\em Jacobian matrix},
  \be 
  \hat w^\beta= \sum_\alpha \frac{\partial y^\beta}{\partial x^\alpha} w^\alpha.
  \ee
\end{defn}
See \Fig{fig:expK*} for an illustration of $\exp_K$, $\exp_{K*}$, and the various tangent spaces used in this paper.

\begin{figure}[h]
\includegraphics[width=\columnwidth]{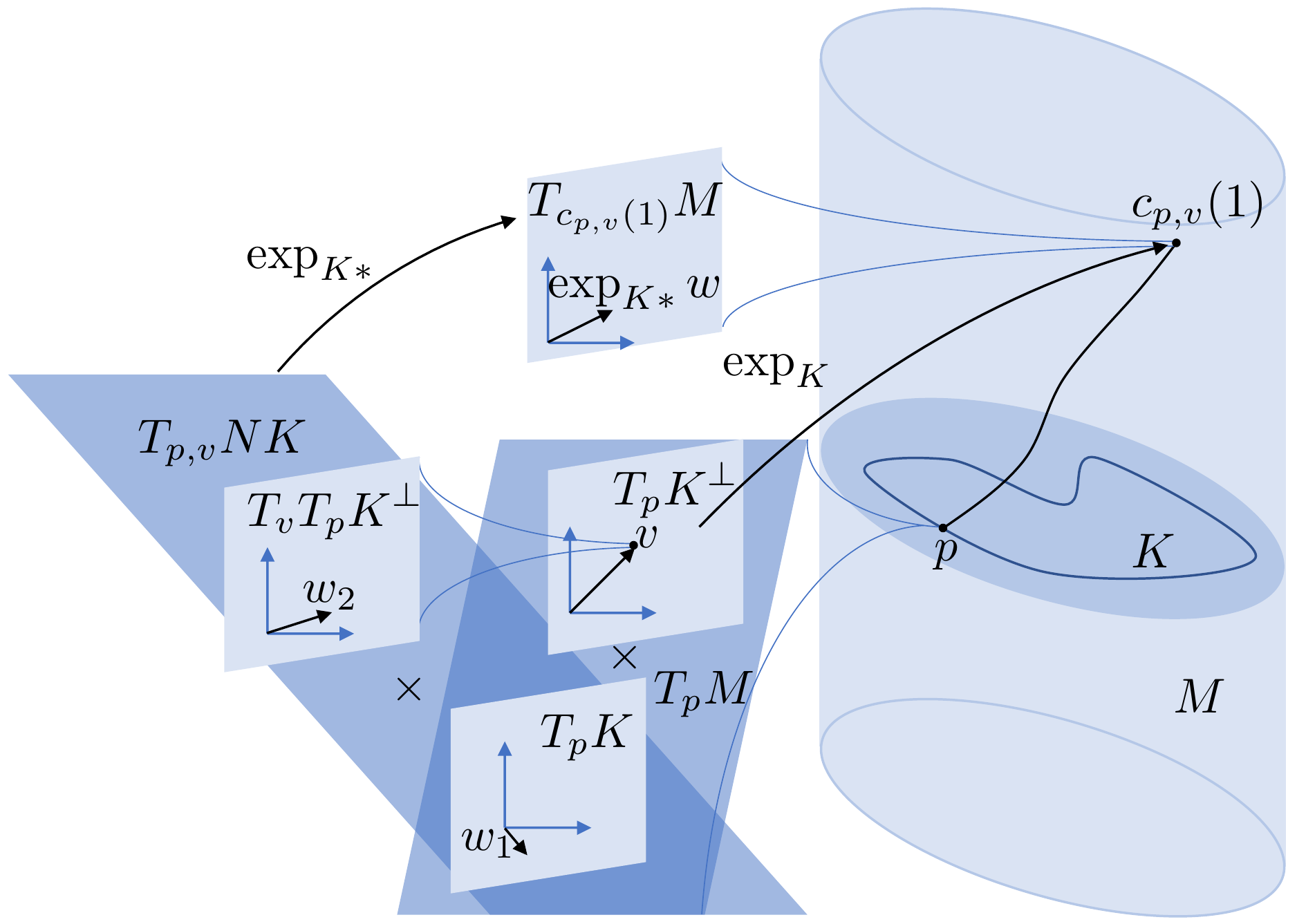}
\caption{An illustration of the surface-orthogonal exponential map $\exp_K$ evaluated at $p\in K$, which takes a vector in $T_p K^\perp$ to a point $c_{p,v}(1)$ in $M$. Here, as in text, the tangent space at $p$, $T_p M$, is broken up as a product $T_p K^\perp \times T_p K$. Also shown is the Jacobian $\exp_{K*}$ at $v\in T_p K^\perp$, which takes a vector $w=(w_1,w_2)\in T_{p,v}NK=T_p K \times T_v T_p K^\perp$ to a vector in $T_{c_{p,v}(1)}M$.}\label{fig:expK*}
\end{figure}

\begin{defn}
A Jacobian is an {\em isomorphism} if it is invertible, i.e., if it has no eigenvectors with eigenvalue zero. 
\end{defn}
Since $(M,g)$ and $K$ are smooth, $\exp_K$ is smooth. The inverse function theorem \cite{rudin} thus implies the following.
\begin{lem}\label{lem-iso}
If the Jacobian $\exp_{K*}$ at $(p,v)\in NK$ is an isomorphism, then $\exp_K$ is a diffeomorphism of an open neighborhood of $(p,v)$ onto an open neighborhood of $\exp_K(p,v)\in M$.
\end{lem}
\begin{defn}
\label{def-cpk}
The exponential map $\exp_K$ is called {\em singular} at $(p,v)\in NK$ if $\exp_{K*}$ is not an isomorphism. Then $(p,v)$ is called a {\em conjugate point} in $NK$. 
\end{defn}

\subsection{Jacobi Fields} \label{sec-jac}

It is instructive to relate the above definition of {\em conjugate point} to an equivalent definition in terms of Jacobi fields. 
\begin{defn}
Let $Q$ be an open set in $\mathbb{R}^2$ and let $f:\,Q\to M,\, (r,s)\mapsto f(r,s)$ be a smooth map. If the curves of constant $r$ and varying $s$, $\gamma_r: Q\to M, s\mapsto f(r,s)$, are geodesics in $M$, then $f$ is called a {\em one-parameter family (or congruence) of geodesics}.
\end{defn}
\begin{defn}
Let $\partial_s$ denote the partial derivative with respect to $s$. It follows from the above definition that the pushforward $S\equiv f_*(\partial_s)\in TM$ is tangent to any geodesic $\gamma_r$. Similarly, $R\equiv f_*(\partial_r)\in TM$ is tangent to any curve $\mu_s:Q\to M, r\mapsto f(r,s)$ at fixed $s$. For general families of curves, $R$ represents the deviation vector field of the congruence. In the special case of a geodesic congruence, $R$ restricted to any $\gamma_r$ is called a {\em Jacobi field} on $\gamma_r$.
\end{defn}
\begin{rem}
The Jacobi field $R$ satisfies the geodesic deviation equation on $Q$, 
\begin{equation}
D^2_S R = \mathcal{R}(S,R) S,
\label{eq-gde}
\end{equation}
where $\mathcal{R}(A,B) \equiv [D_A,D_B]-D_{[A,B]}$ is the curvature tensor~\cite{Hicks,Wald} and $D_V = V^\mu \nabla_\mu$ is the covariant derivative, defined with respect to the Levi-Civita connection, along a vector $V$.
\end{rem}
The exponential map can be used to generate a one-parameter family of geodesics and its derivative $\exp_*$ generates the associated Jacobi fields. We first recall the more familiar case of geodesics through a point $p$, generated by $\exp_p$, as follows.
\begin{rem}
Let $\hat R, \hat S\in T_pM$ and let $\tilde R$ and $\tilde S$ be the naturally associated constant vector fields in $TT_pM$.\footnote{Concretely, one can first choose a neighborhood $U$ of $p$ diffeomorphic to $\mathbb{R}^n$, which exists since $M$ is a manifold, and then choose a map $\phi:U\rightarrow T_p M$ such that the pushforward $\phi_*$ is the identity map from $T_p M$ to $T_v T_p M$ for some $v$; then $\tilde R$ and $\tilde S$ can be defined as $\tilde R=\phi_* \hat R$ and $\tilde S = \phi_* \hat S$ for $v=\hat R$ or $\hat S$, respectively.} Then $f(r,s) = \exp_p[s(\hat S+r\hat R)]$ is smooth and defines a one-parameter family of geodesics. Its tangent vector field is $S=\left.\exp_{p*}\right|_{s(\hat S+r\hat R)} ( \tilde S+r \tilde R)$ and its deviation or Jacobi field is $R=\left.\exp_{p*}\right|_{s(\hat S+r\hat R)} s\tilde R$.\footnote{The subscript is the point where the Jacobian map is evaluated. The vector the Jacobian acts on appears to its right.} It is clear from this construction that $\exp_{p}$ is singular (i.e., $\exp_{p*}$ fails to be an isomorphism) at $s(\hat S+r\hat R)$ if and only if there exists a nontrivial Jacobi field of the geodesic $\gamma_r$ that vanishes at $f(r,s)$ and $f(r,0).$ This establishes the equivalence of two common definitions of conjugacy to a point $p$.
\end{rem}

\begin{rem} A conjugate point in a geodesic congruence with tangent vector $k^\mu$ corresponds to a {\em caustic}, which is a point at which the {\em expansion} $\theta = \nabla_\mu k^\mu$ goes to $-\infty$. 
\end{rem}
We turn to the case relevant to this paper: a one-parameter family of geodesics orthogonal to a smooth, compact, acausal, codimension-two submanifold $K$. (For example, $K$ could be a topological sphere at an instant of time.) Subject to this restriction, the map $f$ and vector fields $R$ and $S$ are defined as before, with $T_p M$ replaced by $T_p K^\perp$. One can choose the parameters $(r,s)$ such that $f(r,0)\in K$ and $f(0,0)=p$. The map $\nu:r\mapsto (f(r,0),S|_{(r,0)})$ is a smooth curve in $NK$ with tangent vector $\bar R\in TNK$. From this curve, the one-parameter family can be recovered as 
\begin{equation}
f(r,s) = \exp_{f(r,0)} sS|_{(r,0)} = \exp_K (f(r,0),sS|_{(r,0)}).
\label{eq-construct}
\end{equation}
\begin{rem}\label{rem:initval}
We will be interested in the Jacobi field $R\equiv f_* \partial_r$ only along one geodesic, say $\gamma$ at $r=0$. By \Eq{eq-gde} this depends only on the initial data $S$ and $\bar R$ at $p$.  Thus $R|_{(0,s)}$ will be the same for any curve $\nu$ with tangent vector $\bar R$ at $(p,S|_{(0,0)})\in NK$. Conversely, one can extend any given $\bar R$ at $(p,S|_{(0,0)})\in NK$ to a (non-unique) one-parameter family of geodesics by picking such a curve $\nu$. We now take advantage of this freedom in order to find an explicit expression for the Jacobi field in terms of $\exp_{K*}$.
\end{rem}

By \Eq{eq-factor}, one can uniquely decompose $\bar R = (\check R ,\tilde R)$, with $\check R\in T_pK$ and $\tilde R\in T_ST_pK^\perp$. Let $\pi$ be the defining projection of the fiber bundle, $\pi: NK\to K$. Then $\mu \equiv \pi(\nu)$ is a curve on $K$ with tangent vector $\check R$ at $p$. Let $f(r,0)=\mu(r)$.  

Further, let $S|_{(r,0)}\in T_{f(r,0)}K^\perp$ be defined by $K$-normal parallel transport\footnote{Given a vector $v\in T_p K^\perp$, normal parallel transport defines a vector field $v(r)$ along $\mu$ normal to $K$ such that the normal component of its covariant derivative along $\mu$ vanishes, $D_r^\perp v(r)=0$. Given $\mu(r)$ and the initial vector in $T_p K^\perp$, $v(r)$ is unique by Lemma 4.40 of \Ref{oneill}.\label{npt}} of the vector $S|_{(0,0)}+r \hat R \in T_pK^\perp$ along $\mu$ from $p$ to $\mu(r)$. Here $\hat R \in T_pK^\perp$ is the vector naturally associated with $\tilde R \in T_S T_p K^\perp.$ Similarly, we define $\tilde S \in T_S T_p K^\perp$ to be the vector naturally associated with $\left.S\right|_{(0,0)}.$
\begin{lem}
\label{lem-recon}
With the above choices and definitions, \Eq{eq-construct} yields a suitable one-parameter family of geodesics. The corresponding Jacobi field and tangent vector along $\gamma$ can be written as:
\begin{equation}
R|_{(0,s)} \equiv \left. f_* \partial_r\right|_{(0,s)} = \left.\exp_{K*}\right|_{(p,sS|_{(0,0)})}(\check R,s \tilde R) \label{eq-r}
\end{equation}
and
\begin{equation}
S|_{(0,s)} \equiv \left. f_* \partial_s \right|_{(0,s)} = \left.\exp_{K*}\right|_{(p,sS|_{(0,0)})}(0, \tilde S), \label{eq-s}
\end{equation}
respectively.
\end{lem}
See App.~\ref{app} for a proof of Lemma~\ref{lem-recon} via a direct calculation.

We note that $\check R$ and $\tilde R$ encode the initial value and derivative, respectively, of $R$, in accordance with the initial value problem set up in Remark~\ref{rem:initval}.
From \Eq{eq-r}, we obtain a criterion for conjugacy equivalent to that of Def.~\ref{def-cpk}:
\begin{rem}
\label{rem-kequiv}
In the above notation, the map $\exp_K$ is singular at $(p,sS|_{(0,0)})\in NK$ if and only if the geodesic $\gamma$ possesses a nontrivial Jacobi field that vanishes at $\exp_K (p,sS|_{(0,0)})$ and is tangent to $K$ at $p.$
\end{rem}
Specifically, our interest lies in null geodesics orthogonal to $K$. We now show that their conjugate points satisfy an additional criterion on the associated eigenvector of $\exp_{K*}$.
\begin{lem} \label{lem-conjtoK}
Let $\gamma$ be a geodesic orthogonal to $K$ at $p$, with conjugate point $(p,sS|_{(0,0)})\in NK$. By Def.~\ref{def-cpk} there exists a nonzero vector $\bar R\in T_{p,sS|_{(0,0)}}NK$ such that $\left. \exp_{K*}\right|_{(p,sS|_{(0,0)})} \bar R=0$. If $\gamma$ is null, i.e., if $\Vert{S|_{(0,0)}}\Vert=0$, then the projection of $\bar R$ onto $T_pK$ is nonvanishing: $\check R \neq 0$. 
\end{lem}
\begin{proof} By Eqs.~\eqref{eq-r} and \eqref{eq-s}, the Jacobi field $R|_{(0,s)}$ is orthogonal to $\gamma$ at two points: at $p$ (by construction) and (trivially) at the assumed conjugate point. By Lemma~8.7 of Ref.~\cite{oneill}, this implies that $R_{(0,s)}\perp S|_{(0,s)}$ for all $s$. Again using Eqs.~(\ref{eq-r}) and ~(\ref{eq-s}), along with linearity of $\exp_{K*}$, this implies that $\tilde R \perp \tilde S$ and thus $\exp_{K*}|_{(p,sS|_{(0,0)})} (0,s \tilde R) \perp S.$

Prior to the conjugate point, the map $\exp_{K*}$ is a linear isomorphism; hence it maps the (1+1)-dimensional subspace $T_ST_pK^\perp\ni \tilde R$ of $T_{p,S}NK$ into a (1+1)-dimensional subspace $\exp_{K*} T_ST_pK^\perp$ of $T_{f(0,1)}M$. This subspace contains both the null tangent vector $S|_{(0,s)}$ and the component $\exp_{K*}|_{(p,sS|_{(0,0)})} (0,s\tilde R)$ of the Jacobi field $R$, which is itself a Jacobi field since our choice of initial data $\bar R$ was arbitrary. In a (1+1)-dimensional space, the only vectors orthogonal to a null vector $S$ are proportional to $S$. The general solution to \Eq{eq-gde} for a Jacobi field proportional to the tangent vector $S$ is $(\alpha + \beta s) S|_{(0,s)}$. Therefore $\exp_{K*}|_{(p,sS|_{(0,0)})} (0,s\tilde R)$ must have this form for some real constants $\alpha,\beta$. At $s=0$, $\exp_{K*}|_{(p,sS|_{(0,0)})}(0,s\tilde R)$ vanishes trivially, so $\alpha = 0$.

Now, suppose $\check R = 0$, so $R|_{(0,s)}$ is just $\beta s S|_{(0,s)}$. Since our Jacobi field is nontrivial and $S$ does not vanish, we must have $\beta \neq 0$. Thus, $R|_{(s,0)}$ vanishes only at $p$ and hence cannot vanish at $\exp_K(p,sS|_{(0,0)})$. This contradiction implies that $\check R \neq 0$.
\end{proof}

We now define a refinement of the notion of a conjugate point.

\begin{defn}\label{def-conjtoK} Let $\gamma(s)$ be a geodesic orthogonal to $K$ at $p$, with $\gamma(0)=p$ and with conjugate point $(p,v)$. Then there exists a nontrivial Jacobi field $R(s)\in TM$ that vanishes at $q=\exp_K (p,v)$ and is tangent to $K$ at $p.$ We say that $q$ is {\it conjugate to (the surface) $K$} if $R$ is nonvanishing at $p.$
\end{defn}

\begin{rem}\label{rmk-conjtoK} By Lemma~\ref{lem-conjtoK}, $\check R \neq 0$, so the Jacobi field associated with $\bar R$ as defined in \Eq{eq-r} does not vanish at $p$ and hence, if $(p, sS|_{(0,0)}) \in NK$ is a conjugate point, then the point $\exp_K(p,sS|_{(0,0)})$ is conjugate to $K$ for $\gamma$ null.
\end{rem}

Moreover, we can similarly define the notion of a point conjugate to another point.

\begin{defn}\label{def-ptpconj}
Given a nontrivial Jacobi field $R$ for a segment $\gamma$ of a geodesic such that $R$ vanishes at $p$ and $q$, we say that $q$ is {\it conjugate to (the point) $p$}.
\end{defn}

See \Fig{fig:Jacobi} for an illustration of the two types of conjugate points defined in Defs.~\ref{def-conjtoK} and \ref{def-ptpconj}.

\begin{figure}[h]
\includegraphics[width=0.7\columnwidth]{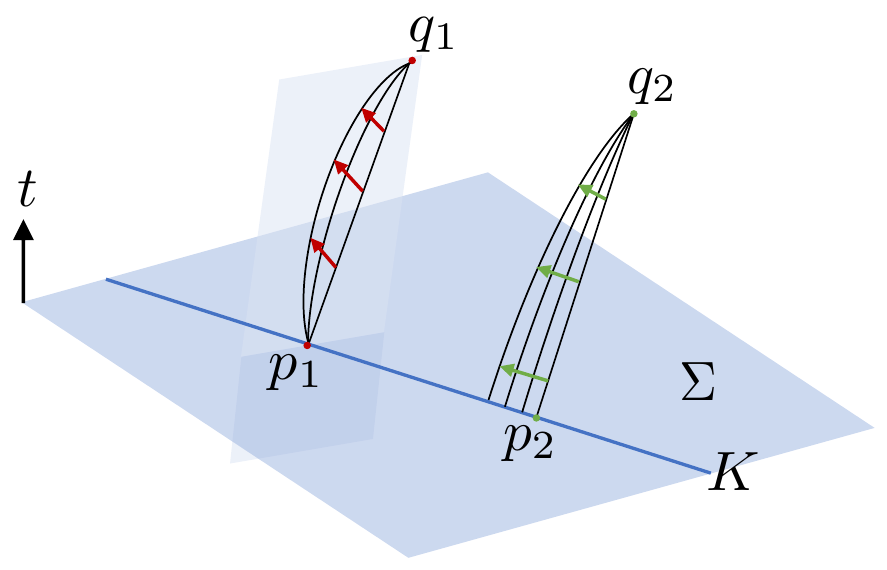}
\caption{The two types of conjugate points defined in Defs.~\ref{def-conjtoK} and \ref{def-ptpconj}. The point $q_1$ is conjugate to the {\em point} $p_1$, with the Jacobi field illustrated by the red arrows. The point $q_2$ is conjugate to the {\em surface} $K$ (blue line), at the point $p_2$, with the Jacobi field illustrated by the green arrows. Geodesics orthogonal to $K$ are shown in black. If a general conjugate point lies along an orthogonal null geodesic, then by Lemma~\ref{lem-conjtoK} there exists a Jacobi field such that the conjugate point is of the {\em surface} type. Hence, this type of conjugacy appears in Theorem~\ref{Theorem}.}\label{fig:Jacobi}
\end{figure}

\section{Proof of the Theorem}
\label{proof}

\renewcommand \qedsymbol{$\blacksquare$}

We now prove Theorem~\ref{Theorem}.
\begin{proof}
For the ``only if'' direction, we may assume that $b\in \dot I^+(K)$. Then conclusions $(i),(ii)$ are already established explicitly elsewhere in the literature (e.g., Theorem~9.3.11 of \Ref{Wald} and Theorem 7.27 of \Ref{Penrose}; see also Lemma~VII of \Ref{Penrose2}, as well as \Ref{Hawking2}). 

Conclusion $(iii)$ follows by contradiction: let $\gamma'$ be a distinct null geodesic orthogonal to $K$ that intersects $\gamma$ at some point $q$ strictly between $b$ and $K$. By acausality of $K$, $\gamma'\cap K$ is a single point, $p'$, which is distinct from $q$. Hence, $K$ can be connected to $b$ by a causal curve that is not an unbroken null geodesic, namely, by following $\gamma'$ from $p'$ to $q$ and $\gamma$ from $q$ to $b$. By Proposition 4.5.10 in Ref.~\cite{HawEll}, this implies that some $r\in K$ can be joined to $b$ by a timelike curve, in contradiction with $b\in \dot I^+(K)$. Hence, no such $\gamma'$ can exist.

The ``if'' direction of the theorem states that if $(i),(ii),(iii)$ hold, then $b\in \dot I^+(K)$. We will prove the following equivalent statement: If $b \notin \dot I^+(K)$ satisfies $(i)$, then $b$ will fail to satisfy $(ii)$ or $(iii)$. 

Let the geodesic $\gamma(s)$ guaranteed by $(i)$ be parametrized so that $\gamma(0)=p\equiv \gamma \cap K$ and $\gamma(1)=b$. By $(i)$, $b\in J^+(K)$, the causal future of $K$. By assumption, $b\notin \dot I^+(K)=\dot J^+(K)$, so it follows that $b\in I^+(K)$, the chronological future of $K$. Since $p\in \dot I^+(K)$, there exists an $s_*$ between 0 and 1 where $\gamma$ leaves the boundary of the future:
\begin{equation}
s_*\equiv\sup \gamma^{-1}(\gamma([0,1])\cap \dot I^+(K)).
\label{eq-sstar}
\end{equation}
The point where $\gamma$ leaves $\dot{I}^+{(K)}$, $q\equiv\gamma(s_*)$, lies in $\dot I^+(K)$.\footnote{This follows because $\dot I^+(K)$ is closed and hence its intersection with a closed segment of $\gamma$ is closed. Therefore, the argument of the supremum is a closed interval and the supremum is its upper endpoint.} Thus $s_*<1$. Moreover, $s_*> 0$ by the obvious generalization of Proposition 4.5.1 in~\Ref{HawEll} and achronality of $K$. We conclude that
\begin{equation}
p\in \dot I^-(q) \cap K,~~q\neq b,~\text{and}~q\neq p.
\end{equation}

Recall that $q=\gamma(s_*)$ is the future-most point on $\gamma$ that is not in $I^+(K)$. Let $s_n$ be a strictly decreasing sequence of real numbers that converges to $s_*$. That is, $s_n>s_*$ and, for $n$ sufficiently large, the points $q_n\equiv\gamma(s_n)$ exist and lie in $I^+(K)$. Now, since $K$ is acausal and $M$ is globally hyperbolic, there exists a Cauchy surface $\Sigma \supset K.$ Given $p_1, p_2 \in M$, define $C(p_1,p_2)$ to be the set of all causal curves from $p_1$ to $p_2$. Since by Corollary~6.6 of \Ref{Penrose} $C(\Sigma,q_n)$ is compact, it is closed and bounded. Thus, $C(K,q_n)\subset C(\Sigma,q_n)$ is bounded. Consider a sequence of curves $\mu_{m}$ from $K$ to $q_n$. By Lemma 6.2.1 of \Ref{HawEll}, the limit curve $\mu$ of $\{\mu_m\}$ is causal; since $K$ is compact and thus contains its limit points, $\mu$ runs from $K$ to $q_n$, so $\mu\in C(K,q_n)$. Hence, $C(K,q_n)$ is closed and therefore compact. Since the proper time is an upper semicontinuous function on $C(\Sigma,q_n)$, it attains its maximum over a compact domain, so we conclude in analogy with Theorem~9.4.5 of \Ref{Wald} that there exists a timelike geodesic $\gamma_n$ that maximizes the proper time from $K$ to $q_n$. By Theorem~9.4.3 of \Ref{Wald}, $\gamma_n$ is orthogonal to $K$. 

By construction, the point $q$ is a convergence point (and hence a limit point) of the sequence $\{\gamma_n\}$. By the time-reverse of Lemma 6.2.1 of Ref.~\cite{HawEll}, there exists, through $q$, a causal limit curve $\gamma'$ of the sequence $\{\gamma_n\}$. This curve must intersect $K$ because all $\gamma_n$ intersect $K$ and $K$ is compact. Since $\gamma'$ passes through $q\in \dot{I}^+(K)$, it must not be smoothly deformable to a timelike curve since $I^+(K)$ is open. Thus, by Theorem~9.3.10 of \Ref{Wald}, $\gamma'$ must be a null geodesic orthogonal to $K$, so if $\gamma'\neq \gamma$, condition $(iii)$ fails to hold. See \Fig{fig:proof} for an illustration.

The only alternative is that $\gamma$ is the only limit curve of the sequence $\{\gamma_n\}$. In this case, $\{\gamma_n\}$ contains a subsequence whose convergence curve is $\gamma$.  From now on, let $\{\gamma_n\}$ denote this subsequence. Orthogonality to $K$ of the $\gamma_n$ implies that we can write
\begin{equation}
  q_n=\exp_K (p_n, v_n),
\end{equation}
where $p_n=\gamma_n\cap K$, for some vector $v_n\in T_{p_n}K^\perp$ tangent to $\gamma_n$. But since $q_n\in \gamma$, we can also write
\begin{equation}
  q_n = \exp_K (p,k_n),
\end{equation}
where $k_n$ is tangent to $\gamma$. Thus, every $q_n$ has a non-unique pre-image.

By the above construction, the sequences $\{(p,k_n)\}$ and $\{(p_n,v_n)\}$ in $NK$ each have $(p,v)$ as their limit point, where $q=\exp_K(p,v)$. Hence there exists no open neighborhood $\cal O$ of $(p,v)$ such that $\exp_K$ is a diffeomorphism of $\cal O$ onto an open neighborhood of $q$. By Lemma~\ref{lem-iso}, it follows that $\exp_K$ is singular at $(p,v)$, i.e., $(p,v)$ is a conjugate point. By Lemma~\ref{lem-conjtoK} and Remark~\ref{rmk-conjtoK}, $q$ is conjugate to $K$. Thus, condition $(ii)$ fails to hold; again, see \Fig{fig:proof}.
\end{proof}

\renewcommand \qedsymbol{$\square$}

\begin{figure}[H]\begin{center}
\includegraphics[width=0.7\columnwidth]{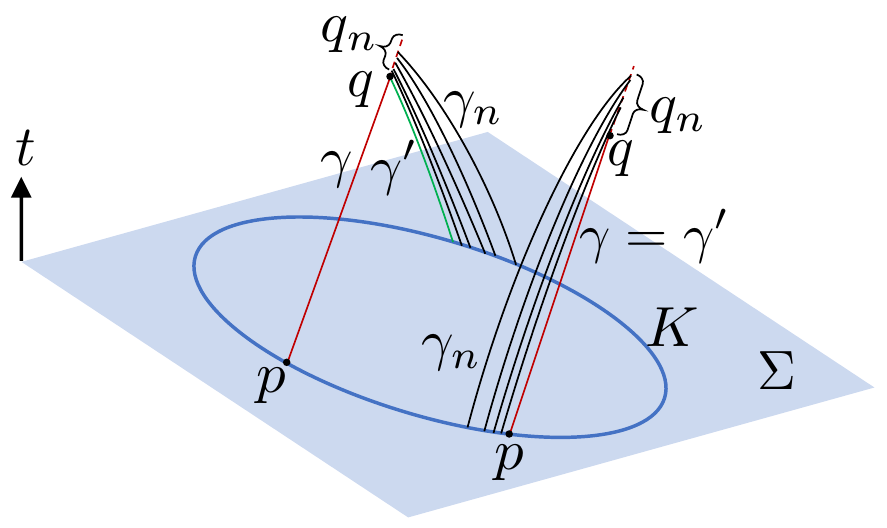}\end{center}
\caption{Possibilities in the proof. The sequence of timelike geodesics $\gamma_n$ (black) connects $K$ with a sequence of points $q_n \in I^+(K)$ on the orthogonal null geodesic $\gamma$ (red) that joins $p\in K$ with $q$, after which $\gamma$ leaves $\dot{I}^+(K)$ (red dashed). In the case on the left, $\gamma'$ (green) is distinct from $\gamma$, so condition~({\it iii}) fails. In the case on the right, $\gamma'=\gamma$, which we prove corresponds to a failure of condition ({\it ii}).}\label{fig:proof}
\end{figure}

\begin{rem}
The fact that $K$ had codimension two was only important in the last step in the proof of Theorem~\ref{Theorem}, i.e., going from knowing that $(p,v)$ is a conjugate point to showing that $q$ is conjugate to the surface $K$. For $K$ a compact, acausal submanifold that is not of codimension two, the steps in the proof of Theorem~\ref{Theorem} still establish that $(p,v)$ is a conjugate point in the sense of Def.~\ref{def-cpk}. Moreover, that the corresponding Jacobi field is orthogonal to $S$ remains true without the codimension-two assumption (see the proof of Lemma~\ref{lem-conjtoK}) and the one-parameter family of geodesics is orthogonal to $K$ (because it was defined via normal parallel transport). As a result, the Jacobi field defines a deviation of $\gamma$ in terms of only orthogonal null geodesics (as proven in, e.g., Corollary 10.40 of \Ref{oneill}), but in general that will not mean that $q$ is conjugate to the surface $K$ in the sense of Def.~\ref{def-conjtoK}. Specifically, the Jacobi field is not necessarily nonvanishing at $K$ if $K$ has codimension greater than two.
\end{rem}

\vskip .3cm
\indent {\bf Acknowledgments} 
It is a pleasure to thank Jason Koeller for initial collaboration on this project. We also thank Netta Engelhardt, Zachary Fisher, Stefan Leichenauer, and Robert Wald for helpful discussions and correspondence. We thank Umberto Lupo for pointing out Refs.~\cite{X,Y} to us after this paper first appeared. This work was supported in part by the Berkeley Center for Theoretical Physics, by the National Science Foundation (award numbers PHY-1521446, PHY-1316783), by FQXi, and by the U.S. Department of Energy under contract DE-AC02-05CH11231. G.N.R.\ is supported by the Miller Institute for Basic Research in Science at the University of California, Berkeley.

\begin{appendix}

\section{Proof of Lemma~\ref{lem-recon}} \label{app}

We now prove Lemma~\ref{lem-recon} by direct calculation.

\begin{proof}

We wish to show that
\be 
R|_{(0,s)} \equiv  \left. f_* \partial_r\right|_{(0,s)} = \left.\exp_{K*}\right|_{(p,sS|_{(0,0)})}(\check R,s \tilde R) \label{eq-ra}
\ee
and
\be 
S|_{(0,s)} \equiv  \left. f_* \partial_s \right|_{(0,s)} = \left.\exp_{K*}\right|_{(p,sS|_{(0,0)})}(0, \tilde S),\label{eq-sa}
\ee
where
\begin{equation}
\begin{aligned}
f(r,s) &= \exp_{f(r,0)} sS|_{(r,0)} \\&= \exp_K (f(r,0),sS|_{(r,0)}),\label{eq:appf}
\end{aligned}
\end{equation}
as defined in \Sec{sec-jac}. 

Using the definition of the pushforward, we can write $f_* \partial_r|_{(0,s)}$ as the differential $\exp_{K*}$, associated with $f$ in \Eq{eq:appf}, evaluated along the tangent direction $sS|_{(0,0)}$,
\begin{equation}
\begin{aligned}
R|_{(0,s)} \equiv& \left. f_* \partial_r\right|_{(0,s)}  \\
=& \left.\exp_{K*}\right|_{(p,sS|_{(0,0)})}(\check R, s\phi_*(\partial_r S|_{(r,0)})|_{r=0}).
\label{eq-ra2}
\end{aligned}
\end{equation}
In the second line, we used the definition of $\check R$ as the tangent to $\mu(r)$ at $p$, along with linearity of $\exp_{K*}$. We have again used the notation $\phi_*$ for the identity map between vectors in $T_p M$ and their naturally associated counterparts in $T_S T_p M$.

Next, we must evaluate the derivative of $S$, $\phi_*(\partial_r S|_{(r,0)})|_{r=0}\in T_S T_p K^\perp$. Let us write $S|_{(r,0)}$ as an explicit function of both the parameter $r$ along the path $\mu(r)\equiv f(r,0)\in K$ and the vector $S|_{(0,0)}+r \hat R\in T_p K^\perp$ that is normal parallel transported along $\mu$ from $\mu(0)=p$ to $\mu(r)$:
\begin{equation}
\begin{aligned}
S|_{(r,0)} &= S(r, S|_{(0,0)}+r \hat R)|_{r_1 = r_2 = r} \\&\equiv S(r_1,r_2)|_{r_1 = r_2 = r},
\end{aligned}
\end{equation} 
so that the derivative in question can be written as $\phi_*\partial_r S(r, S|_{(0,0)}+r \hat R)|_{r = 0}$. Since $S(r_1,r_2)$ is defined by normal parallel transporting a particular vector ($S|_{(0,0)}+r \hat R$) in $T_p K^\perp$ to $\mu(r_1)$, its variation with respect to $r_1$ gives the normal part of the covariant derivative of $S$ along $\mu$, which vanishes, i.e., $\partial_{r_1} S(r_1,r_2) = 0$. Hence,
\be 
\begin{aligned}
&	\frac{\partial}{\partial r} \left[S(r_1, S|_{(0,0)}+r_2 \hat R)|_{r_1 = r_2 = r}\right] \\&= \left[\frac{\partial}{\partial r_2} S(r,r_2)\right]_{r_2 = r} \\ & = \hat R.
	\end{aligned}
\ee
Inputting this result into \Eq{eq-ra2}, we have
\be
\begin{aligned}
R|_{(0,s)} 
=& \left.\exp_{K*}\right|_{(p,sS|_{(0,0)})}(\check R, s \phi_* \hat R)\\
=& \left.\exp_{K*}\right|_{(p,sS|_{(0,0)})}(\check R,s \tilde R).
\label{eq-ra3}
\end{aligned}
\ee
We have thus derived the claimed formula for the Jacobi field stated in \Eq{eq-ra}. The proof of \Eq{eq-sa} follows similarly. Neither $f(r,0)$ or $S|_{(r,0)}$ depend on $s$. Therefore
\be 
\begin{aligned}
S|_{(0,s)} \equiv& \left. f_* \partial_s\right|_{(0,s)}  \\
=& \partial_s \exp_K(f(0,0), s S|_{(0,0)})\\
=& \left.\exp_{K*}\right|_{(p,sS|_{(0,0)})}(0, \phi_* S|_{(0,0)}) \\
=& \left.\exp_{K*}\right|_{(p,sS|_{(0,0)})}(0, \tilde S).
\end{aligned}
\ee
This derivation of the Jacobi field and tangent vector completes the proof of Lemma~\ref{lem-recon}.
\end{proof}

\end{appendix} 

\newpage

\bibliography{Boundary}
\end{document}